%% file: elastic_minimizer.tex
\documentclass[a4paper]{amsart}
\usepackage[T1]{fontenc}
\usepackage[utf8]{inputenc}
\usepackage[italian, english]{babel}

\usepackage{amsmath}
\usepackage{amssymb}
\usepackage{amsthm}
\usepackage{hyperref}
\usepackage{subfig}
\usepackage{graphicx}
\usepackage{amsaddr}
\usepackage{tabularx}
\usepackage{booktabs} 

\newtheorem{thm}{Theorem}[section]

\newtheorem{lem}{Lemma}[section]
\newtheorem{prop}{Proposition}[section]
\newtheorem{dfn}{Definition}[section]
\newtheorem{ax}{Axiom}[section]
\newtheorem{rem}{Remark}[section]

\graphicspath{{fig/}}
\usepackage{xcolor}
\hypersetup{
    colorlinks,
    linkcolor={red!50!black},
    citecolor={blue!50!black},
    urlcolor={blue!80!black}
}

\input{defs}

\author{D. Riccobelli, A. Agosti and P. Ciarletta}
\address{MOX -- Dipartimento di Matematica, Politecnico di Milano, Milano, Italy.}
\email{davide.riccobelli@sissa.it}
\email{abramo.agosti@polimi.it}
\email{pasquale.ciarletta@polimi.it}
\title[Existence of elastic minimizers for initially stressed materials]{On the existence of elastic minimizers for initially stressed materials}
\begin{document}
\begin{abstract}
A soft solid is said to be initially stressed if it is subjected to a state of internal stress in its unloaded reference configuration.

\noindent Developing a sound mathematical framework to model initially stressed solids in nonlinear elasticity is key for many applications in engineering and biology. This work investigates the links between the existence of elastic minimizers and the constitutive restrictions for initially stressed materials subjected to finite deformations. In particular,  we consider a subclass of constitutive responses in which the strain energy density is taken as a scalar valued function of both the deformation gradient and the initial stress tensor. The main advantage of this approach is that the initial stress tensor belongs to the group of the divergence-free symmetric tensors satisfying the boundary condition in any given reference configuration. However, it is still unclear which physical restrictions must be imposed for the well-posedness of this elastic problem. Assuming that the constitutive response depends on the choice of the reference configuration only through the initial stress tensor, under given conditions we prove the local existence of a relaxed state given by an implicit tensor function of the initial stress distribution. This tensor function is generally not unique, and can be transformed accordingly to the symmetry group of the material at fixed initial stresses. These results allow to extend Ball's existence theorem of elastic minimizers for the proposed  constitutive choice of initially stressed materials.
\end{abstract}
\maketitle

\section{Rivlin's legacy on constitutive equations in nonlinear physics}
Ronald Rivlin made seminal contribution for the development of constitutive theories in continuum physics.

As collected by Barenblatt and Joseph \cite{rivlin1997collected}, Rivlin wrote 31 papers on isotropic finite elasticity, 8 on the anisotropic theory of elasticity and more than 50 papers on the theory of constitutive equations.

Pioneering contributions include the rigorous mathematical theory of isotropic \cite{rivlin1997large} and anisotropic \cite{ericksen1954large, smith1957stress} nonlinear elasticity, that shaped our modern approach to the formulation of constitutive laws in continuum mechanics \cite{rivlin1955stress, rivlin1955further}.

Moreover,  Pipkin and Rivlin \cite{pipkin1959formulation} studied the constitutive restrictions for enforcing the invariance of the physical law by rotating simultaneously both the actual and the reference system, with particular interest in nonlinear elasticity \cite{rivlin1987note}. Rivlin's work also focused on the physical implication of the material symmetries \cite{rivlin1997some} and the well-posedness of the non-linear elastic problems \cite{rivlin1948large, rivlin1948uniqueness}.

Inspired by his works, in this paper we aim at studying the role of physical restrictions  for the development of a constitutive theory of  initially stressed bodies, namely elastic media that are described exploiting a reference configuration that does not coincide with the relaxed one. 

\section{Introduction to initially stressed materials}

 Existence theorems for nonlinear elastic materials have been developed over the last decades on a different mathematical ground than the ones for linear elasticity, that are classically based on the Korn's inequality \cite{fichera1973existence,horgan1995korn,ciarlet2005another}. The required objectivity of the strain energy function of a soft materials indeed conflicts with its convexity \cite{coleman1959thermostatics}. Thus, less restrictive conditions for the existence of the finite elastic solution are imposed using a direct approach in the calculus of  variations, such as quasi-convexity \cite{morrey1952quasi}, quasi-convexity at the border \cite{ball1984quasiconvexity} and polyconvexity \cite{ball1976convexity}. Without attempting an exhaustive mathematical characterization of such seminal results, we emphasize that proving the existence of elastic minimizers is interwoven with the search for physical restrictions on the elastic strain energy function, that classically contains a nonlinear dependence only on the deformation gradient $\tens{F}$.

 \noindent This work investigates the link between the existence of elastic minimizers and the constitutive assumptions for initially stressed materials subjected to finite deformations.  If the unloaded reference configuration is not unstressed, we say that the elastic body is subjected to an initial stress  $\tens{\Sigma}$, that  must satisfy the equilibrium equations.
 Initial stresses are commonly observed in soft materials. In inert matter they can be actively controlled by an external stimulus, e.g. in hydrogels \cite{nardinocchi2017swelling} and in dielectric elastomers \cite{brochu2010advances}. Such stimuli, e.g. an electric field in dielectric elastomers, generate a distortion of the microstructure of the material that is made physically compatible by the emergence of an internal state of stress. In living matter, initial stresses are also known as residual stresses \cite{Hoger1985,johnson1993dependence,johnson1995use,johnson1998use}, and they result from incompatible
 growth processes both in healthy and pathological conditions \cite{stylianopoulos2012causes, ciarletta2016residual}. Such residual stresses not only may enhance the functionality and the efficiency of biological structures, e.g. in  arteries \cite{chuong1986residual},  but they may also be used to trigger a programmed shape transition through a mechanical instability, forming complex patterns such as the intestinal villi \cite{ciarletta2014pattern} or the brain sulci \cite{bayly2014mechanical}.

From a constitutive viewpoint, a well established approach to account for initial stresses  is based on the multiplicative decomposition of the deformation gradient into an elastic deformation tensor and an incompatible tensor \cite{bilby1955continuous,kroner1959allgemeine, Lee1969}. This method was initially applied to provide a kinematic description of crystal plasticity \cite{reina2014kinematic}, and later adapted to describe the volumetric growth in nonlinear elastic media \cite{rodriguez1994stress}.
Assuming a material isomorphism for the strain energy function, the initial stress is constitutively related to the elastic deformation tensor from the virtual incompatible state. Thus, it has been shown that the resulting elasticity tensor constitutively depends on the initial stress  \cite{johnson1993dependence,johnson1995use,johnson1998use}, that in turn may affect the symmetry group of the material  \cite{johnson1995use}. Since the incompatible tensor is not necessarily the gradient of a deformation, it maps the unloaded configuration into a virtual state that may not possess a Euclidean metric \cite{johnson1998use}. Accordingly, the main drawback of this approach is that such a virtual state may not be achieved in physical practice, not even by cutting procedures, and  the incompatible tensor must be assumed a priori in order to generate a self-equilibrated state of initial stresses.

A less restrictive mathematical framework accounts for initial stresses by formulating implicit constitutive equations linking the Helmoltz free energy, the initial stress and the kinematic quantities possibly mapping the evolving natural states of the materials \cite{rajagopal2003implicit}. For soft solids, this constitutive approach has  shown that there exists a far richer class of non-dissipative materials than the class of bodies that is usually understood as being elastic \cite{rajagopal2007response}.

\noindent In this work, we consider a subclass of constitutive responses in which the strain energy function is taken as a scalar valued function of both the deformation gradient and of the initial stress. As first discussed in \cite{shams2011initial}, objectivity is enforced for an initially stressed material made by an isotropic material by considering  a dependence on the ten invariants of the two above mentioned tensors. Under the incompressibility constraint, it has been shown that only eight invariants are independent \cite{shariff2017spectral}. This method has been widely used to model initially stressed materials; applications of this theory include wave propagation in soft media \cite{shams2012rayleigh, ogden2011propagation}, the modeling of residual stress in living tissues \cite{wang2014modified} and the stability of residually stressed materials \cite{ciarletta2016morphology, rodriguez2016helical, riccobelli2017shape}.

\noindent The main advantage of this approach is that the initial stress tensor $\tens{\Sigma}$ belongs to the group of the divergence-free symmetric tensors satisfying the boundary condition in the given reference configuration, whilst it is still unclear which physical restrictions must be imposed for the well-posedness of the elastic problem. A basic constitutive restriction  known as the initial stress compatibility condition (ISCC) imposes that the Cauchy stress reduces to the initial stress when the deformation tensor is equal to the matrix identity \cite{shams2011initial, gower2015initial}. By imposing ISCC and the polyconvexity of the resulting strain energy function in the absence of initial stresses few constitutive relations have been proposed.  A simple functional expression has been proposed in \cite{merodio2016extension}, containing material parameters that also depend on the particular choice of the reference configuration, as generally prescribed by  \cite{truesdell1965non}. A more restrictive constitutive class has been proposed in \cite{gower2015initial}, assuming that the material parameters do not change under a change of reference configuration. This assumption has lead to define a new condition. i.e. the initial stress reference independence (ISRI)  \cite{gower2016new}, that is inspired by the  multiplicative decomposition approach.

\noindent This work aims at clarifying some constitutive aspects of the mathematical theory of initially stressed materials, unraveling the main implications of imposing the ISRI condition on the existence of elastic minimizers.

\noindent The article is organized as follows. In Section 3, we provide some  basic kinematic and constitutive notions for nonlinear elastic materials. In Section 4, we introduce the main differences of the proposed mathematical framework for initially stressed materials with respect to the theory of elastic distortions, discussing the mechanical signification of the ISCC and the ISRI conditions. In Section 5, we prove the local existence of a relaxed state for each material point. In Section 6, we prove that the residual stresses provoke an elastic distortion on the transformation of the symmetry group. We also give en existence theorem  for the elastic minimizers for the proposed  constitutive choice of the initially stressed material.  In Section 7, we use the proposed framework to solve the physical problem of an elastic disc subjected to an anisotropic initial stress. Finally, the results are summarized and critically discussed in the last section.

\section{Background and notation}

\begin{table}[t]
\caption{List of functional spaces}
\label{table_example}
\begin{tabularx}{\textwidth}{cX}
\toprule
Symbol &Definition \\
\midrule
$\L(\R^n)$ &Set of all the linear applications from $\R^n$ to $\R^n$.\\
$\L^+(\R^n)$ &Set of all the $\tens{L}\in\L(\R^n)$ such that $\det\tens{L}>0$.\\
$\O(\R^n)$ &Set of all the orthogonal tensors $\tens{Q}\in\L(\R^n)$, namely all the $\tens{Q}$ such that $\tens{Q}^T\tens{Q}=\tens{I}$.\\
$\O^+(\R^n)$ &Set of all $\tens{Q}\in\mathcal{O}(\R^n)$ such that $\det\tens{Q}>0$.\\
$\S(\R^n)$ &Set of all the symmetric linear applications from $\R^n$ to $\R^n$, namely all the $\tens{L}\in\L(\R^n)$ such that $\tens{L}^T=\tens{L}$\\
$\L^+_1(\R^n)$ &Special unitary group, namely the subset of $\L^+(\R^n)$ with determinant $1$.\\
$\mathcal{D}$ &Denotes  $\L^+(\R^n)$ for a compressible material,  or  $\L^+_1(\R^n)$ for an incompressible material.\\
$C^0(U,\,V)$ &Set of all the continuous function from the set $U\subseteq\R^n$ to the set $V\subseteq\R^N$.\\
$C^k(U,\,V)$ &Set of all the function from the set $U\subseteq\R^n$ to the set $V\subseteq\R^N$ admitting continuous derivatives of order $k$.\\
$L^p(U,\,V)$ &Set of all the function from the set $U\subseteq\R^n$ to the set $V\subseteq\R^N$ with finite $L^p$~norm.\\
$W^{1,p}(U,\,V)$ &Sobolev space of all the functions from the set $U\subseteq\R^n$ to the set $V\subseteq\R^N$, where both the functions and their weak partial derivatives belong to $L^p(U,\,V)$.\\\bottomrule
\end{tabularx}
\vspace*{-4pt}
\end{table}

We denote by $\L(\R^n)$ the set of all the automorphisms of $\R^n$, and with $\L^+(\R^n)$ the group (with respect to the operation of function composition) of all the linear applications belonging to $\L(\R^n)$ with positive determinant.

Let $\O(\R^n)$ be the group such that
\[
\tens{Q}^T\tens{Q}=\tens{I}
\]
where $\tens{Q}\in\L(\R^n)$ and $\tens{I}$ is the identity.

We indicate with $\O^+(\R^n)\subset\O(\R^n)$ the group of all the elements of $\O(\R^n)$ with positive determinant; if $n=3$, this group coincides with the set of the rigid rotations. We also introduce the set $\S(\R^n)$ of all the symmetric linear applications that belong to $\L(\R^n)$.

Let the open set $\Omega_0\subset\R^3$ be the reference configuration of a body and $\vect{X}\in\Omega_0$ the material point. We denote the deformation field by $\vect{\varphi}\in C^2(\Omega_0,\,\R^3)$that maps the reference domain $\Omega_0$ to the actual configuration $\Omega$. 

Accordingly, the deformation gradient reads $\tens{F}=\Grad\vect{\varphi}$. If the body is made of a homogeneous elastic material, we assume a purely elastic constitutive behavior such that  the Cauchy stress tensor $\tens{T}_0$ depends on the deformation gradient $\tens{F}$.

We say that the body has a relaxed reference configuration if
\begin{equation}
\label{eq:relaxed}
\tens{T}_0(\tens{I})=\tens{0},
\end{equation}
where $\tens{T}_0$ is the Cauchy stress and $\tens{I}$ is the identity tensor. 

If the body is composed of a hyperelastic material, we denote its strain energy density in a point $\vect{X}$ with 
$\psi_0(\tens{F}(\vect{X})):\L^+(\R^3)\rightarrow\R$.
Whenever appropriate, we omit the explicit dependence of the physical quantities on the material position $\vect{X}$. The first Piola--Kirchhoff and the Cauchy stress tensor are  given by
\[
\tens{P}_0(\tens{F})=\frac{\partial\psi_0}{\partial\tens{F}}\qquad\tens{T}_0(\tens{F})=\frac{1}{\det\tens{F}}\tens{P}_0(\tens{F})\tens{F}^T,
\]
respectively.

In order to account for an incompressibility constraint, we introduce the following group:
\[
\mathcal{L}^+_\delta(\R^3)=\left\{\tens{F}\in\L^+(\R^n) \;|\;\det\tens{F}=\delta\right\}.
\]

Accordingly, the domain of the strain energy density $\psi_0$ is given by $\mathcal{L}^+_1(\R^3)$, where the argument is the special unitary group. However, it is convenient to introduce an extension of $\psi_0$ to all $\L^+(\R^3)$ and then to use the method of Lagrangian multiplier to enforce the incompressibility constraint. Let $\tilde{\psi}_0:\L^+(\R^3)\rightarrow \R$ such that
\begin{equation}
\label{eq:tildepsi0}
\tilde{\psi}_{0}(\tens{F})=\psi_0(\tens{F})\qquad\forall \tens{F}\in\L^+_1(\R^3),
\end{equation}
a possible extension is given by
\[
\tilde{\psi}_0(\tens{F})=\psi_0((\det \tens{F})^{-1/3}\tens{F}).
\]

So, the first Piola--Kirchhoff and the Cauchy stress tensors are given by
\[
\tens{P}_0(\tens{F},\,p)=\frac{\partial\tilde{\psi}_0}{\partial\tens{F}}-p\tens{F}^{-T}\qquad\tens{T}_0(\tens{F},\,p)=\tens{P}_0(\tens{F},\,p)\tens{F}^T,
\]
where $p$ is the Lagrangian multiplier. For the sake of simplicity, in the
following we will omit the distinction between $\tilde{\psi}_0$ and $\psi_0$  wherever appropriate and we denote by $\mathcal{D}$ either the group $\L^+(\R^3)$, if the material is unconstrained, or the group $\mathcal{L}^+_1(\R^3)$, if the material is incompressible.


We denote by
\begin{equation}
\label{eq:prin_inv}
I_1(\tens{C})=\tr(\tens{C}),\qquad I_2(\tens{C})=\frac{(\tr\tens{C})^2-\tr(\tens{C}^2)}{2},\qquad I_3(\tens{C})=\det(\tens{C}),
\end{equation}
the principal invariants of $\tens{C}$, where $\tens{C}=\tens{F}^T\tens{F}$ is the right Cauchy--Green strain tensor.
%

We finally introduce a fundamental notion for the existence of minimizers in nonlinear elastic materials, known as the non-degeneracy axiom  \cite{ball1976convexity}:
\begin{ax}[Non-degeneracy for a  hyperelastic body]
\label{ax:nondeg}
Let $\psi_0$ be a strain energy density, we say that $\psi_0$ is non-degenerate if
\begin{equation}
\label{eq:non_deg}
\left\{
\begin{aligned}
&\psi_{0}(\tens{F})\rightarrow +\infty&&\text{\rm when }\det\tens{F}\rightarrow 0^+\\
&\psi_{0}(\tens{F})\rightarrow +\infty&&\text{\rm when }\left|\tens{F}\right|+\left|\tens{F}^{-1}\right|\rightarrow +\infty
\end{aligned}
\right.
\end{equation}
where $|\tens{F}|=\sqrt{\tr(\tens{F}^T\tens{F})}$.
\end{ax}

The last condition of \eqref{eq:non_deg} indeed ensures that the hyperelastic energy goes to infinity as soon as one of the principal invariants \eqref{eq:prin_inv} goes to $+\infty$. If the material is incompressible, only the second equation of \eqref{eq:non_deg} applies. 
For the ease of the readers, we collect all the symbols used to denote the functional spaces in Table~\ref{table_example}.

\section{Mathematical frameworks for initially stressed materials}

In this section, we summarize the basic features of two mathematical frameworks used to model nonlinear elastic materials whose unloaded reference configuration is not stress-free, namely the theory of elastic distortions and the theory of initially stressed bodies.

\subsection{The theory of elastic distortions}
\label{sec:distortions}
If the relation \eqref{eq:relaxed} does not hold, the material is subjected to a state of stress in the reference configuration.
A classical constitutive approach consists in assuming a  multiplicative decomposition of the deformation gradient \cite{rodriguez1994stress}, such that:
\begin{equation}
\label{eq:Fdeco}
\tens{F}=\tens{F}_\text{e}\tens{G}.
\end{equation}
where $\tens{G}$ is the tensor field that describes the elastic distortion from the reference configuration to the relaxed one, whilst $\tens{F}_\mathrm{e}$ represents the elastic distortion that restores the geometrical compatibility under the action of external tractions (as depicted in Fig.~\ref{fig:FG}). Since the underlying metric is not Euclidean  whence $\tens{G}$ is not a gradient of a  deformation field, it may be impossible to attain a  stress-free configuration in the physical world.
In the last decades, the distortion tensor $\tens{G}$  has been advocated to model different biological processes,  such as volumetric growth \cite{dicarlo2002growth}, remodelling \cite{epstein2000thermomechanics} and active strains \cite{kondaurov1987finite, taber2000modeling}.

\begin{figure}
\centering
\includegraphics[width=0.7\textwidth]{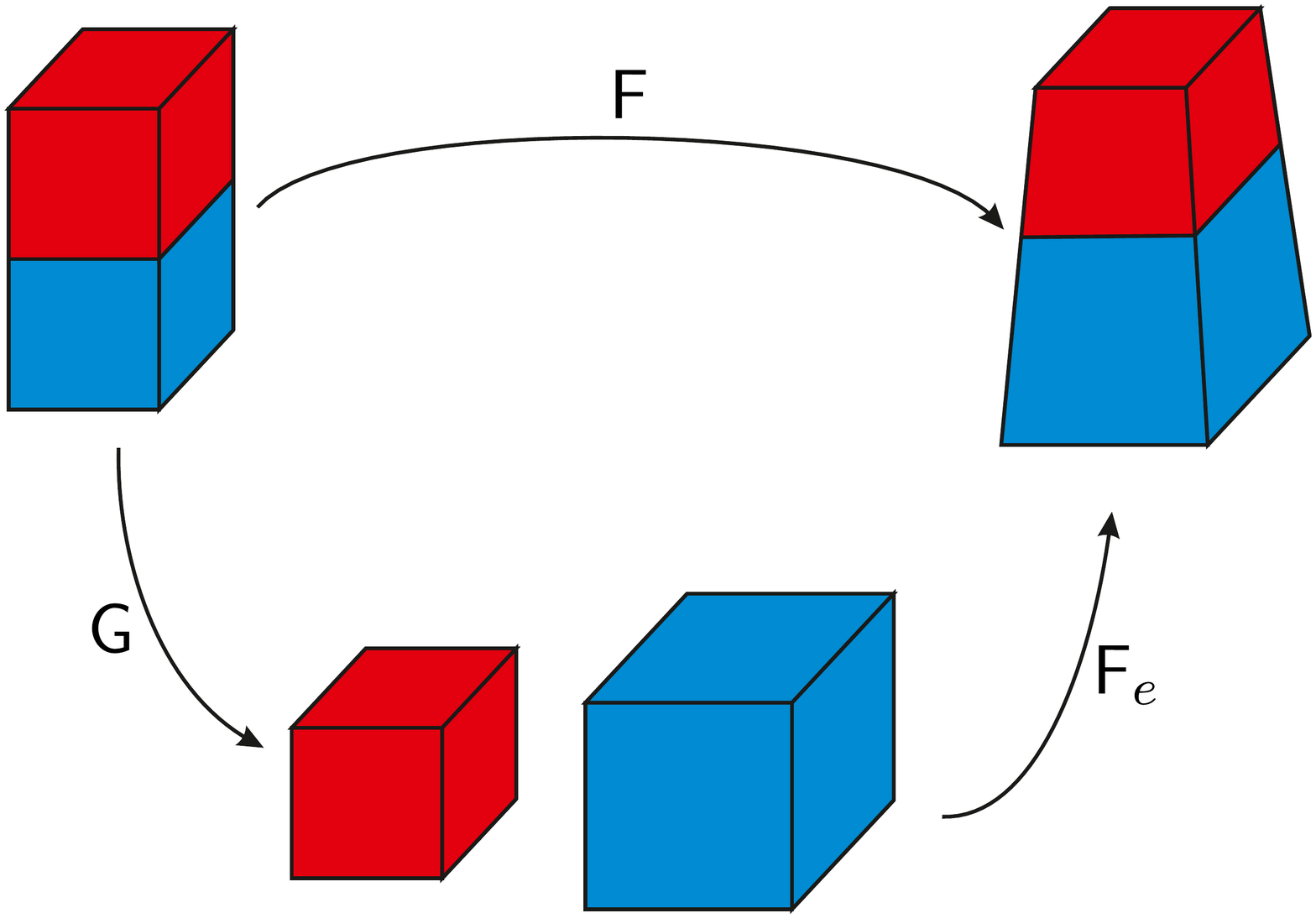}
\caption{Clockwise representation of the reference, the actual and the relaxed configuration described by  the multiplicative decomposition of the deformation gradient given by \eqref{eq:Fdeco}.}
\label{fig:FG}
\end{figure}

In physical practice, it is assumed the initial stress in the body is generated by a distortion of the reference configuration. Consequently, the strain energy function depends on the distorted metric, given by $\tens{F}\tens{G}^{-1}$.

If the material is incompressible, this constraint is imposed on the elastic tensor, whilst the distortion tensor also describes the  local change of volume, such that:
\[
\det\tens{F}_\mathrm{e} = 1\quad\Rightarrow\quad\det\tens{F} = \det\tens{G}=\delta.
\]

Accordingly, the strain energy density of the material is given by :
\begin{equation}
\label{eq:ini_strain}
\psi_\tens{G}(\tens{F})=(\det\tens{G})\psi_0(\tens{F}\tens{G}^{-1}).
\end{equation}
From standard application of the second law of the thermodynamics in the Clausius-Duhem form, the first Piola--Kirchhoff and Cauchy stress read
\begin{equation}
\label{eq:initial_strain}
\tens{P}_\tens{G}(\tens{F})=(\det\tens{G})\frac{\partial\psi_0(\tens{F}\tens{G}^{-1})}{\partial\tens{F}},\qquad\tens{T}_{\tens{G}}(\tens{F})=\frac{1}{\det\tens{F}}\tens{P}_{\tens{G}}(\tens{F})\tens{F}^T.
\end{equation}

%
%
%
%

The theory of distortions provides a transparent explanation for the transformation law of the material properties. Let $\mathcal{G}$ be  the material symmetry group of a hyperelastic material, it is defined as the set of all the tensors $\tens{Q}\in\mathcal{L}^+_1(\R^3)$ such that
\begin{equation}
\label{eq:mat_sym}
\tens{T}(\tens{F})=\tens{T}(\tens{F}\tens{Q}),\qquad\forall\tens{F}\in\mathcal{D};
\end{equation}
where the response function $\tens{T}$ may eventually depend on the local distortion $\tens{G}$.

An equivalent definition for a hyperelastic material can be given as the set of all $\tens{Q}\in\mathcal{L}^+_1(\R^3)$ such that
\[
\psi(\tens{F})=\psi(\tens{F}\tens{Q}),\qquad\forall\tens{F}\in\mathcal{D}.
\]

If we exploit the theory of elastic distortions, let $\mathcal{G}$ be the material symmetry group of the corresponding strain energy density  $\psi_0$. It has been shown that \cite{epstein2015mathematical}:
\begin{equation}
\label{eq:sym_group_dist}
\begin{aligned}
\psi_{\tens{G}}(\tens{F})&=(\det\tens{G})\psi_{0}(\tens{F}\tens{G}^{-1}) =\\
&=(\det\tens{G})\psi_{0}(\tens{F}\tens{G}^{-1}\tens{Q})=\\
&=(\det\tens{G})\psi_{0}(\tens{F}\tens{G}^{-1}\tens{Q}\tens{G}\tens{G}^{-1})=\\
&=\psi_{\tens{G}}(\tens{F}\tens{G}^{-1}\tens{Q}\tens{G})
\end{aligned}
\qquad \forall\tens{Q}\in\mathcal{G}_{0}
\end{equation}
Thus, the material symmetry group of the initially stressed material is given by
\[
\mathcal{G}_{\tens{G}} =  \tens{G}^{-1}\mathcal{G}_{0}\tens{G}.
\]
Notably,  $\mathcal{G}_{\tens{G}}$ is the conjugate group of $\mathcal{G}_{0}$ through $\tens{G}$.

The main drawback of the theory of elastic distortions is that  $\tens{G}$ has to be provided by means of a constitutive assumption. Nonetheless, since the underlying metric may not be Euclidean, the values of its components cannot be directly inferred in many physical problems. An experimental attempt to search for a stress--free configuration consists in performing several (ideally infinite) cuts in the body to release the local stresses stored inside the material \cite{ambrosi2012active, stylianopoulos2012causes, ciarletta2016residual, chuong1986residual, ambrosi2017solid}. Although successful in simple system models \cite{amar2010swelling, moulton2011circumferential, ciarletta2014pattern}, this approach is unsuitable when interested in investigating the effect of a generic state of initial stress on the material response.
In the following, we describe how this difficulty can be circumvented by building a constitutive theory that explicitly depends on the underlying spatial distribution of internal stresses.

\subsection{The theory of initially stressed bodies}

Alternatively, it can be assumed that the material response  depends on both the deformation applied on the body and on the initial stress, intended as  the existing stress field $\tens{\Sigma}$ in the undeformed reference configuration, i.e. the Cauchy stress when the body is undeformed. This assumption has been discussed in  \cite{johnson1993dependence,johnson1995use, johnson1998use}, such that the material response in a point $\vect{X}$ of the body reads
\begin{equation}
\label{eq:resp_funct}
\tens{T}=\tens{T}(\tens{F};\,\tens{\Sigma}(\vect{X})),
\end{equation}
where $\tens{T}$ is the Cauchy stress. We remark that the initial stress tensor field
\[
\tens{\Sigma}:\Omega_0\rightarrow\S(\R^3)
\]
generally depends on the material position vector $\vect{X}$; we omit such an explicit notation in the following for the sake of brevity wherever appropriate. We denote by $\tens{S}\in\S(\R^3)$ the specific expression of the initial stress in a point $\vect{X}$, namely $\tens{S}=\tens{\Sigma}(\vect{X})$ for a given $\vect{X}\in\Omega_0$.


The function $\tens{T}:\mathcal{D}\times\mathcal{S}(\mathbb{R}^3)\rightarrow\mathcal{S}(\mathbb{R}^3)$ must satisfy certain restrictions. First, in the absence of initial stresses, the strain energy function must obey the standard requirements ensuring the existence of elastic minimizers in nonlinear elasticity. Second, the constitutive response should be such that the Cauchy stress is equal to the initial stress in the absence of any elastic deformations. This is referred to as ISCC, i.e. \emph{initial stress compatibility condition}, \cite{shams2011initial}  and reads:
\begin{equation}
\label{eqn:ISCC}
\tens{T}(\tens{I};\,\tens{S})=\tens{S}\qquad  \forall \tens{S} \in \S(\R^3).
\end{equation}

 A subclass of material responses in which the strain energy function depends only on the elastic deformation and the initial stress, but not explicitly on the choice of the reference configuration, has been proposed in \cite{gower2015initial,gower2016new}. Under this constitutive assumption, it is possible to introduce another restriction called Initial Stress Reference Independence (ISRI), stating  that
\begin{equation}
\label{eq:ISRIT}
\tens{T}(\tens{F}_2\tens{F}_1;\,\tens{S})=\tens{T}\left(\tens{F}_2;\,\tens{T}(\tens{F}_1;\,\tens{S})\right), \quad \forall \tens{S} \in \S(\R^3).
\end{equation}

In this work we give a new mechanical interpretation of such a restrictive condition and we discuss its mathematical implications for the existence of elastic minimizers. The condition \eqref{eq:ISRIT} imposes that  there is no energy dissipation resulting from the elastic deformation and represents a frame invariance requirement: the deformation field solution of the elastic problem must not depend on the choice of the reference configuration \cite{gower2015initial, gower2016new}.

If the material is hyperelastic, we can assume that the strain energy function reads \cite{shams2011initial}:
\begin{equation}
\label{eq:ini_stress_ener}
\psi:\L^+(\R^3)\times\S(\R^3)\rightarrow\R.
\end{equation}

Recalling the values of the material parameters are assumed to be independent on the choice of the initially stressed configuration, the  first Piola--Kirchhoff  tensor $\tens{P}$ and the Cauchy stress tensor $\tens{T}$ read:
\begin{equation}
\label{eq:stresses}
\tens{P}(\tens{F};\,\tens{S})=\frac{\partial \psi}{\partial \tens{F}}(\tens{F};\,\tens{S}),\qquad\tens{T}(\tens{F};\,\tens{S})=\frac{1}{\det \tens{F}}\tens{P}(\tens{F};\,\tens{S})\tens{F}^T,\qquad \tens{F}\in\mathcal{L}^+(\R^3), \quad \tens{S}\in\S(\R^3).
\end{equation}

Under the incompressibility constraint, the strain energy density is a function such that
\[
\psi:\L^+_1(\R^3)\times\S(\R^3)\rightarrow\R,
\]
as done in \eqref{eq:tildepsi0}, we introduce an extension $\tilde{\psi}$ of $\psi$ to all $\L^+(\R^3)$ to define the stress tensors, namely
\begin{equation}
\label{eq:psitilde}
\tilde{\psi}(\tens{F};\,\tens{S})=\psi(\tens{F};\,\tens{S})\qquad\forall\tens{F}\in\L^+_1(\R^3), \quad \forall\tens{S}\in\S(\R^3);
\end{equation}
a possible extension is given by:
\[
\tilde{\psi}(\tens{F};\,\tens{S})=\psi((\det\tens{F})^{-1/3}\tens{F};\,\tens{S}).
\]
The Piola--Kirchhoff and the Cauchy stress tensors are given by
\[
\left\{
\begin{aligned}
&\tens{P}(\tens{F},\,p;\,\tens{S})=\frac{\partial \tilde{\psi}}{\partial \tens{F}}(\tens{F};\,\tens{S})-p\tens{F}^{-T},\\
&\tens{T}(\tens{F},\,p;\,\tens{S})=\tens{P}(\tens{F},\,p;\,\tens{S})\tens{F}^T,
\end{aligned}
\right.\quad \tens{F}\in\mathcal{L}^+_1(\R^3),\quad \tens{S}\in\S(\R^3).\]

For the sake of simplicity, we omit the difference between $\tilde{\psi}$ and $\psi$ wherever appropriate.
For hyperelastic materials, we remind that \eqref{eq:ISRIT} can be reformulated as  an equivalent condition to be imposed on the functional dependence of the strain energy function \cite{gower2016new}. We give further mathematical details of this important result in the following, proving that the restriction imposed on the strain energy density is a consequence of \eqref{eq:ISRIT}.
\begin{prop}
Let $\psi:\mathcal{D}\times\S(\mathbb{R}^3)\rightarrow\R$ be a strain energy density, and assume that the ISCC \eqref{eqn:ISCC} and the ISRI conditions \eqref{eq:ISRIT} hold. Then, for all $\tens{F}_1,\,\tens{F}_2\in\mathcal{D}$ and for all $\tens{S}\in\S(\R^3)$, the following relation must hold:
\begin{equation}
\label{eq:ISRIW}
\psi(\tens{F}_2\tens{F}_1;\,\tens{S})=(\det{\tens{F}_1})\psi(\tens{F}_2;\,\tens{T}(\tens{F}_1;\,\tens{S})).
\end{equation}
\end{prop}
\begin{proof}
For the sake of brevity, let $\psi$ be the strain energy of a compressible material (the incompressible case is analogous). The ISRI condition \eqref{eq:ISRIT} reads
\begin{equation}
\label{eq:thm1}
\frac{1}{\det{\tens{F}}_1\det{\tens{F}}_2}\frac{\partial\psi}{\partial\tens{F}}(\tens{F}_2\tens{F}_1;\,\tens{S})\tens{F}_1^{T}\tens{F}_2^{T}=\frac{1}{\det{\tens{F}}_2}\frac{\partial\psi}{\partial\tens{F}_2}(\tens{F}_2;\,\tens{T}(\tens{F}_1;\,\tens{S}))\tens{F}_2^{T}.
\end{equation}

Since $\tens{F}=\tens{F}_2\tens{F}_1$, then, by using the chain rule, we obtain
\[
\frac{\partial \psi}{\partial \tens{F}_2}=\frac{\partial \psi}{\partial\tens{F}}\tens{F}_1^T
\]

and the equation \eqref{eq:thm1} becomes
\[
\frac{\partial\psi}{\partial\tens{F}}(\tens{F}_2\tens{F}_1;\,\tens{S})=\det{\tens F}_1\frac{\partial\psi (\tens{F}\tens{F}_1^{-1},\,\tens{T}(\tens{F}_1;\,\tens{S}))}{\partial\tens{F}}.
\]
We find that $\psi(\tens{F}_2\tens{F}_1;\,\tens{S})=(\det{\tens{F}_1})\psi(\tens{F}_2;\,\tens{T}(\tens{F}_1;\,\tens{S}))+C$. Setting $\tens{F}_1=\tens{I}$ and making use of \eqref{eqn:ISCC} we find  that $C=0$ and we get the claim.
\end{proof}



In the next sections, we prove the local existence of a relaxed state around each material point and a theorem on the existence of elastic minimizers for a strain energy of the form given by \eqref{eq:ini_stress_ener}.

\section{Existence of a relaxed state}

In the theory of elastic distortions, we must  provide a constitutive form for the tensor field that we should apply locally to each point in the reference configuration to obtain the (virtual) relaxed one \cite{rodriguez1994stress}. This theoretical framework has strong mathematical properties. Indeed, if $\psi_0$ is polyconvex, then also $\psi_\tens{G}$ (defined in \eqref{eq:ini_strain}) inherits such a property \cite{neff}. As discussed earlier, this approach is straightforward but only suitable in simple system models, since it requires the a priori knowledge of the virtual relaxed state. 

In this section we prove a theorem on the existence of relaxed configuration using the constitutive framework of initially stressed bodies. Moreover, we prove that, if a strain energy satisfy the ISRI \eqref{eq:ISRIW} and $\psi(\cdot,\tens{0})$ is polyconvex, then $\psi(\cdot,\,\tens{S})$ is polyconvex for all $\tens{S}\in\S(\R^3)$.

First, we give  the following  statement of the non-degeneracy axiom for this class of materials.
\begin{ax}[Non-degeneracy for an initially stressed body]
\label{ax:nondeg2}
Let $\psi:\L^+(\R^3)\times\S(\R^3)\rightarrow\R$ be the strain energy density of an initially stressed body. We say that $\psi$ is non-degenerate if
\begin{equation}
\label{eq:non_deg_stressed}
\left\{
\begin{aligned}
&\psi(\tens{F};\,\tens{S})\rightarrow +\infty&&\text{when }\det\tens{F}\rightarrow 0^+,\\
&\psi(\tens{F};\,\tens{S})\rightarrow +\infty&&\text{when }\left|\tens{F}\right|+\left|\tens{F}^{-1}\right|\rightarrow +\infty.
\end{aligned}
\right.\qquad\forall \tens{S}\in \mathcal{S}(\R^3).
\end{equation}
\end{ax}

We now derive the existence of a point-like relaxed state associated to each point of the initially stressed configuration.  In \cite{johnson1995use}, a stress-free virtual state is defined for each material point in the initially stressed configuration by considering the limiting  behavior  as the radius of the spherical neighborhood tends to zero. Its existence required the following  hypotheses:  
$\tens{\Sigma}\in C^1(\Omega_0,\mathcal{S}(\R^3))$, $\psi=\psi(\tens{F};\,\tens{\Sigma}(\vect{X}))$ being twice differentiable with respect to both arguments, and the distortion from the neighborhood of each point to the free state to be once differentiable in space. Here we are going to obtain a proof of existence of a virtual state using weaker hypotheses.

\begin{thm}[Existence of a relaxed state]
\label{thm:existence}
Let $\psi$ be a non-degenerate strain energy density in the sense of \eqref{eq:non_deg_stressed}. We also assume that $\psi(\cdot,\tens{S})$ is at least $C^1$ and proper, i.e. it is not identically equal to $+\infty$, for all $\tens{S}\in\S(\R^3)$.

Then, for each $\vect{X} \in \Omega_0$, given $\tens{\Sigma}(\vect{X})\in\S(\R^3)$, there exists a local distortion $\tens{G}_{\tens{\Sigma}(\vect{X})}$ such that 
\[
\tens{T}(\tens{G}_{\tens{\Sigma}(\vect{X})};\,\tens{\Sigma}(\vect{X}))=\tens{0}.
\]
\end{thm}
\begin{proof}
Let $\psi$ be the strain energy of a compressible material.

We denote by $f_{\tens{\Sigma}(\vect{X})}=\psi(\cdot;\,\tens{\Sigma}(\vect{X}))$. The domain of the function $f_{\tens{\Sigma}(\vect{X})}$ is given by $\mathcal{D}=\mathcal{L}^+(\R^3)$. Thus, from the non-degenericity axiom \eqref{ax:nondeg2}, we get
\begin{equation}
\label{eq:fdomain}
f_{\tens{\Sigma}(\vect{X})}(\tens{F})\rightarrow +\infty\qquad \text{ when } \det\tens{F}\rightarrow 0^+ \text{ or }|\tens F|\rightarrow +\infty.
\end{equation}

Since the function $f_{\tens{\Sigma}(\vect{X})}$ is continuous and proper, it must be bounded from below, hence there exists a value $m\in\R$ such that $f_{\tens{\Sigma}(\vect{X})}(\tens{F})> m$ for all $\tens{F}$. Moreover, there exists a value $M>m$ such that
\begin{equation}
\label{eq:setU}
f_{\tens{\Sigma}(\vect{X})}^{-1}\left([m,\,M]\cap \Imag(f_{\tens{\Sigma}(\vect{X})})\right)=U\neq\emptyset,\qquad U\subset\mathcal{L}^+(\R^3).
\end{equation}
where with $f_{\tens{\Sigma}(\vect{X})}^{-1}(A)$ we denote the pre--image of the subset $A\subseteq \Imag(f_{\tens{\Sigma}(\vect{X})})$ through the function $f_{\tens{\Sigma}(\vect{X})}$.

The non-empty set $U$ is bounded as a direct consequence of the coercivity property expressed in \eqref{eq:non_deg_stressed}. 
Thus, there exists a minimum of $f_{\tens{\Sigma}(\vect{X})}$ in $\bar{U}$, where $\bar{U}$ is the closure of $U$.

The tensor that realizes such a minimum may be not unique as exposed in the following Remark \ref{rem:uni}. Let us denote with $\tens{G}_{\tens{\Sigma}(\vect{X})}$ one of them. Since \eqref{eq:fdomain} holds, the tensors $\tens{G}_{\tens{\Sigma}(\vect{X})}$ cannot belong to the boundary of the set $\mathcal{L}^+(\R^3)$ and it is a critical point for $\psi(\cdot;\,\tens{\Sigma}(\vect{X}))$.

From \eqref{eq:stresses}, we get
\begin{equation}
\label{eqn:firstordercondition}
\tens{T}(\tens{G}_{\tens{\Sigma}(\vect{X})};\,\tens{\Sigma}(\vect{X}))=\tens{0}.
\end{equation}

If the material is incompressible, following the same argument, there exists a tensor $\tens{G}_{\tens{\Sigma}(\vect{X})}$ such that
\[
\tens{G}_{\tens{\Sigma}(\vect{X})}\in\argmin_{\tens{F}\in \mathcal{D}}f_{\tens{\Sigma}(\vect{X})}(\tens{F})
\]
where in this case $\mathcal{D}=\mathcal{L}^+_1(\R^3)$. We define a new function $\hat{\psi}$ such that
\[
\left\{
\begin{aligned}
&\hat\psi:\L^+(\R^3)\times\R\times\mathcal{S}(\R^3)\rightarrow\R\\
&\hat\psi(\tens{F},\,p;\,\tens{S})=\tilde{\psi}(\tens{F};\,\tens{S})-p(\det\tens{F}-1)
\end{aligned}
\right.
\]
where $\tilde{\psi}$ is an extension of $\psi$ as defined in \eqref{eq:psitilde}.

Since $\tens{G}_{\tens{\Sigma}(\vect{X})}$ is a minimum for $\psi$ in $\mathcal{L}^+_1(\R^3)$, there exists a $p_{\tens{\Sigma}(\vect{X})}\in\R$ such that $(\tens{G}_{\tens{\Sigma}(\vect{X})}, p_{\tens{\Sigma}(\vect{X})})$ is a critical point for $\hat\psi$ \cite{ekeland1999convex}, so that
\[
\frac{\partial\hat\psi}{\partial\tens{F}}(\tens{G}_{\tens{\Sigma}(\vect{X})},\,p_{\tens{\Sigma}(\vect{X})};\,\tens{\Sigma}(\vect{X}))=\frac{\partial\tilde \psi}{\partial\tens{F}}(\tens{G}_{\tens{\Sigma}(\vect{X})};\,\tens{\Sigma}(\vect{X}))-p_{\tens{\Sigma}(\vect{X})}\tens{G}_{\tens{\Sigma}(\vect{X})}^{-T}=\tens{0},
\]
and thus $\tens{T}(\tens{G}_{\tens{\Sigma}(\vect{X})},\,p_{\tens{\Sigma}(\vect{X})};\,\tens{\Sigma}(\vect{X}))=\tens{0}$. This concludes the proof.
\end{proof}

\begin{rem}
\label{rem:uni}
Given an initial stress tensor, this Theorem implies that there exists a tensor $\tens{G}_{\tens{\Sigma}(\vect{X})}$ that locally maps the body to an unstressed state. Such a distortion is not unique in general: if $\tens{Q}$ belongs to the material symmetry group of $\psi$, then also $\tens{T}(\tens{G}_{\tens{\Sigma}(\vect{X})}\tens{Q};\,\tens{\Sigma}(\vect{X}))=\tens{0}$.
\end{rem}

\begin{rem}
\label{rem:localbound}
The collection of local maps
\[
\widehat{\tens{G}}[\tens{\Sigma}](\vect{X}):=\tens{G}_{\tens{\Sigma}(\vect{X})},
\]
which transform each point of the reference configuration into a point in the local unstressed virtual state, satisfies $\widehat{\tens{G}}[\tens{\Sigma}]\in \mathcal{B}(\Omega_0,\,\mathcal{D})$, where $\mathcal{B}(\Omega_0,\,\mathcal{D})$ denotes the set of all the bounded function $f:\Omega_0\rightarrow \mathcal{D}$. Moreover, the tensor map $\widehat{\tens{G}}[\tens{\Sigma}]$ may not be geometrically compatible, i.e. there could not exist any differentiable vector field $\vect{\varphi}_{\widehat{\tens{G}}}$ such that $\Grad\vect{\varphi}_{\widehat{\tens{G}}} = \widehat{\tens{G}}[\tens{\Sigma}]$. In this case,  there does not exist a deformation that maps the reference configuration of the  residually stressed material into a relaxed one. In fact, assuming that the reference configuration is simply connected, such a deformation exists if and only if
\[
\rot\widehat{\tens{G}}[\tens{\Sigma}] = \mathbf{0}.
\]
In the following, we call $\widehat{\tens{G}}[\tens{\Sigma}]$ the \textit{relaxing map}.
\end{rem}

\begin{rem}
\label{rem:homogeneous}
By simple application of the mean stress theorem \cite{gurtin1973linear}, in the absence of surface tractions and body forces we obtain
\[
\frac{1}{|\Omega_0|}\int_{\Omega_0}\tens{\Sigma}\,d\vect{X} = \frac{1}{|\Omega_0|}\left(\int_{\partial\Omega_0}\vect{X}\otimes(\tens{\Sigma}\vect{N})\,dS -\int_{\Omega_0}\vect{X}\otimes\Diver\tens{\Sigma}\,d\vect{X}\right) = \tens{0}
\]
so that the mean value of the initial stress tensor is zero. Thus, the Cartesian components of the residual stress tensor are necessarily  spatially inhomogeneous whenever $\tens{\Sigma}\neq\tens{0}$ \cite{hoger1986determination}. Accordingly, the functional form of the map $\widehat{\tens{G}}[\tens{\Sigma}]$ is also inhomogeneous. 
A homogeneous initial stress $\tens{\Sigma}$ can only exist if surface tractions or body forces are applied.
\end{rem}

\begin{rem}
\label{rem:zeromeasure}
Note that in the case in which $\tens{\Sigma}$ has singular values over a set $S_{\infty}\subset \Omega_0$ we are requiring the hypothesis that $\psi(\cdot;\,\tens{\Sigma}(\vect{X}))$ remains a proper function, i.e. it is not identically equal to $+\infty$, when $\vect{X}\in S_{\infty}$. Due to the continuity of $\psi(\cdot,\tens{S})$, this ensures that $\widehat{\tens{G}}[\tens{\Sigma}](\vect{X})$ is bounded when $\vect{X}\in S_{\infty}$, as it will be shown in Section \ref{sec:7}.
\end{rem}

%

\section{Existence of elastic minimizers for initially stressed bodies}

In this section, we prove that if the strain energy density $\psi$ satisfies the assumption of Theorem \ref{thm:existence}, then $\psi$ satisfies the ISRI \eqref{eq:ISRIW} if and only if it is expressible using the theory of elastic distortion \eqref{eq:ini_strain}.

\begin{thm}
\label{thm:uniqueness}
Let $\psi$  satisfy the hypotheses of Theorem \ref{thm:existence} and the ISCC condition. We denote by $\psi_0$ the strain energy of the material in the absence of initial stresses, being
\[
\psi_0(\tens{F})=\psi(\tens{F};\,\tens{0}).
\]

Then, the function $\psi$ satisfy the ISRI \eqref{eq:ISRIW} if and only if we can express it as
 \[
\psi(\tens{F};\,\tens{\Sigma}(\vect{X}))=(\det\widehat{\tens G}[\tens{\Sigma}](\vect{X}))\psi_0(\tens{F}\widehat{\tens{G}}[\tens{\Sigma}](\vect{X})^{-1})=\psi_{\widehat{\tens{G}}[\tens{\Sigma}](\vect{X})}(\tens{F})
\]
where $\widehat{\tens{G}}[\tens{\Sigma}](\vect{X})$ is a function such that
\[
\tens{T}(\widehat{\tens{G}}[\tens{\Sigma}](\vect{X});\,\tens{\Sigma}(\vect{X}))=\tens{0}.
\]
\end{thm}
\begin{proof}
It is proved in \cite{gower2016new} that the strain energy $\psi_{\widehat{\tens{G}}[\tens{\Sigma}](\vect{X})}$ satisfies the ISRI.

Let $\psi$ be a strain energy which satisfies the ISRI and such that $\psi(\tens{F};\,\tens{0})=\psi_0(\tens{F})$. The existence of the function $\widehat{\tens{G}}[\tens{\Sigma}](\vect{X})$ is guaranteed by the Theorem \ref{thm:existence}.

Omitting the explicit dependence on $\vect{X}$ for the sake of compactness,  we obtain:
\begin{align*}
\psi(\tens{F};\,\tens{\Sigma})&=\psi(\tens{F}\widehat{\tens{G}}[\tens{\Sigma}]^{-1}\widehat{\tens{G}}[\tens{\Sigma}];\,\tens{\Sigma})=\\
&=(\det\widehat{\tens{G}}[\tens{\Sigma}])\psi(\tens{F}\widehat{\tens{G}}[\tens{\Sigma}]^{-1},\,\tens{T}(\widehat{\tens{G}}[\tens{\Sigma}];\,\tens{\Sigma}))=\\
&=(\det\widehat{\tens{G}}[\tens{\Sigma}])\psi(\tens{F}\widehat{\tens{G}}[\tens{\Sigma}]^{-1},\,\tens{0})=\\
&=(\det\widehat{\tens{G}}[\tens{\Sigma}])\psi_0(\tens{F}\widehat{\tens{G}}[\tens{\Sigma}]^{-1})=\psi_{\widehat{\tens{G}}[\tens{\Sigma}]}(\tens{F})
\end{align*}
that concludes the proof.
\end{proof}

Let us now introduce the following Definiton:

\begin{dfn}
\label{dfn:poly}
A strain energy density  $\psi_0(\tens{F})$ is said \emph{polyconvex} if there exists a convex function $h:\R^{19}\rightarrow\R$ such that
\[
\psi_0(\tens{F})=h(\tens{F},\,\cof \tens{F},\,\det\tens{F}),
\]
 for all $\tens{F}\in\L^+(\R^3).$
\end{dfn}

We also introduce the following useful Lemma:
\begin{lem}
\label{lem:poly}
Let $\psi(\tens{F};\,\tens{\Sigma}(\vect{X}))$ be a strain energy density satisfying the hypotheses of Theorem~\ref{thm:existence}, the ISRI and such that $\psi(\tens{F};\,\tens{0})$ is polyconvex. Then $\psi(\tens{F};\,\tens{\Sigma}(\vect{X}))$ is polyconvex for all $\vect{X}\in\Omega_0$.
\end{lem}
\begin{proof}
From the Theorem~\ref{thm:uniqueness}, we get
\[
\psi(\tens{F};\,\tens{\Sigma}(\vect{X}))=(\det{\widehat{\tens G}}[\tens{\Sigma}](\vect{X}))\psi(\tens{F}\widehat{\tens{G}}[\tens{\Sigma}](\vect{X})^{-1};\,0).
\]
Following the Remark \ref{rem:localbound}, we have that $\widehat{\tens{G}}[\tens{\Sigma}]\in \mathcal{B}(\Omega_0,\mathcal{L}^+(\R^3))$ and the Lemma is a direct consequence of the Lemma $6.5$ in \cite{neff}.
\end{proof}

Indeed, under some regularity assumptions, if the strain energy density $\psi$ is polyconvex in the relaxed case, then it is polyconvex for all $\tens{\Sigma}:\Omega_0 \to \S(\R^3)$. It is now possible to prove a theorem of existence of elastic  minimizers for initially stressed bodies.
\begin{thm}[Existence of elastic  minimizers]
\label{thm:existence_solution}
Let $\Omega_0\subset\R^3$ be a connected, bounded and open subset with a regular boundary and let $\psi(\tens{F};\,\tens{\Sigma}(\vect{X}))$ be a strain energy density for an initially stressed material, with $\psi(\cdot;\,\tens{\Sigma}(\vect{X}))\in C^1(\L^+(\R^3))$ and $\psi(\tens{F};\cdot)\in C^0(\S(\R^3))$. Let $\tens{\Sigma}:\Omega_0 \to \S(\R^3)$ be a measurable function.

We assume that:
\begin{itemize}
\item[(i)] (initial stress independence and non--degeneracy) $\psi$ fulfills the hypotheses of Theorem \ref{thm:existence} and the ISRI condition \eqref{eq:ISRIW};

\item[(ii)](polyconvexity of the relaxed energy) in the absence of initial stresses, the strain energy density $\psi(\tens{F};\,\tens{0})$ is polyconvex with respect to $\tens{F}$, namely there exists a convex function $h:\L(\R^3)\times\L(\R^3)\times(0,\,+\infty)\rightarrow\R$ such that
\[
\psi(\tens{F};\,\tens{0})=h(\tens{F},\,\cof\tens{F},\,\det\tens{F})
\]

\item[(iii)] (coercivity of the relaxed energy) there exist $\alpha >0$, $\beta \in \R$, $p \ge 2$, $q\ge p/(p-1)$, $r>1$ such that:
\[
h(\tens{F},\,\tens{C},\,\delta)\geq \alpha(|\tens{F}|^p+|\tens{C}|^q+\delta^r)+\beta,\qquad\forall\tens{F},\,\tens{C}\in\L(\R^3),\,\delta>0.
\]
\end{itemize}

\noindent We assume that there exist two disjointed subset $\Gamma_0,\,\Gamma_1$ such that $\partial\Omega_0=\Gamma_0\cup\Gamma_1$ and such that $|\Gamma_0|>0$. Let $\vect{f}:\Omega_0\rightarrow \R^3$ and $\vect{t}:\Gamma_1\rightarrow\R^3$ measurable such that the application
\[
L[\vect{\varphi}]=\int_{\Omega_0}\vect{f}\cdot\vect{\varphi}\,d\vect{X}+\int_{\Gamma_1}\vect{t}\cdot\vect{\varphi}\,dS
\]
is continuous on $W^{1,p}(\Omega_0,\,\R^3)$. Finally let $\vect{\varphi}_0:\Gamma_0\rightarrow\R^3$ be a measurable function and such that the set
\begin{multline}
U = \Big\{\vect{\varphi}\in W^{1,p}(\Omega_0,\,\R^3)\;|\;\cof\Grad\vect{\varphi}\in L^q,\,\det\Grad\vect{\varphi}\in L^r,\\
\det\Grad\vect{\varphi}>0\text{ a.e. in }\Omega_0,\,\vect{\varphi}=\vect{\varphi}_0\text{ on }\Gamma_0\Big\}.
\end{multline}
is non-empty.

Then, defining the functional $\mathcal{F}:U\rightarrow\R\cup\{+\infty\}$ as
\[
\mathcal{F}[\vect{\varphi}]=\int_{\Omega_0}\psi(\Grad\vect{\varphi};\,\tens{\Sigma}(\vect{X}))d\vect{X}-L[\vect{\varphi}]
\]
and assuming that $\inf \mathcal{F}[\vect{\varphi}]<+\infty$, there exists an elastic minimizer
\[
\min_{\vect{\varphi}\in U}\mathcal{F}[\vect{\varphi}].
\]
\end{thm}
\begin{proof}
Using  the Theorem \ref{thm:uniqueness},  from (i) we have that
\begin{equation}
\psi(\tens{F};\,\tens{\Sigma}(\vect{X}))=(\det{\widehat{\tens G}}[\tens{\Sigma}](\vect{X}))\psi_0(\tens{F}\widehat{\tens{G}}[\tens{\Sigma}](\vect{X})^{-1}) =: \Psi(\vect{X},\,\tens{F}).
\label{p1}
\end{equation}
We prove the  claim as a direct application of the Theorem 7.3 in \cite{ball1976convexity}. Here we only sketch the proof, pointing to \cite{ball1976convexity} for the details. 

Since $\tens{\Sigma}$ is measurable and $\psi(\tens{F},\cdot)$ is continuous for all $\tens{F}$, then $\Psi(\vect{X},\tens{F})$ is a Carath\'eodory function, i.e. it is continuous with respect to  $\tens{F}$ a.e. in $\Omega_0$ and measurable in $\Omega_0$ for all  $\tens{F}\in\L^+(\R^3)$. Hence, the functional $\mathcal{F}$ is well defined.

By simple application of  Lemma \ref{lem:poly} and (ii), $\Psi(\vect{X},\,\tens{F})$ is polyconvex a.e. in $\Omega_0$. The coercivity of $\mathcal{F}$ is  enforced a.e. in $\Omega_0$ by the hypothesis (iii), the boundedness of $\widehat{\tens{G}}[\tens{\Sigma}]$ in  Remark \ref{rem:localbound} and the continuity of $L$ in $W^{1,p}(\Omega_0,\,\R^3)$. The non-degeneracy of $\Psi(\vect{X},\,\tens{F})$ for $\det \tens{F} \rightarrow 0^+$ is given by (i).
Hence, by applying the standard methods of the calculus of variations, we can show the existence of infimizing sequences $\vect{\varphi}_k \in U$ which admit weakly converging subsequences to a limit point $\vect{\varphi} \in U$. Since the functional $\mathcal{F}$ is lower semicontinuous as a consequence of its policonvexity, the weak limit $\vect{\varphi} \in U$ minimizes $\mathcal{F}$.

\end{proof}
Such a Theorem is a standard application of Ball's theorem on the existence of solutions in nonlinear elasticity \cite{ball1976convexity}. The main result obtained in this section is that \emph{the ISRI automatically guarantees that the polyconvexity is preserved for all the initial stress fields if it holds for $\tens{\Sigma}=\tens{0}$}. Conversely, if we do not assume the ISRI, the polyconvexity of the strain energy density should be imposed by a suitable constitutive restriction on the dependence with respect to the initial stress field.

According to the Theorem \ref{thm:uniqueness}, imposing the ISRI condition is equivalent to require that {the initial stress tensor} $\tens{\Sigma}$ is generated by an elastic distortion given by  $\widehat{\tens{G}}[\tens{\Sigma}]$. 

Conversely, if the ISRI does not hold, the dependence of the stored elastic energy on the choice of the reference configuration is not solely related to the the corresponding variation of the initial stress. Thus, the material properties may depend on the specific initial stress field.

For the sake of clarity, let us investigate how the material symmetry group depends on the presence of an initial stress within the body. We denote by $\mathcal{G}_{0}$ the material symmetry group of the relaxed state around a material point $\vect{X}$. In view of Theorem \ref{thm:uniqueness} and following the same computation of \eqref{eq:sym_group_dist}, if the strain energy fulfills the ISRI condition for a generic initial stress field $\tens{S}$ and for all $\tens{Q}\in\mathcal{G}_{0}$, we get
that the material symmetry group $\mathcal{G}_\tens{S}$ of the initially stressed body is given by
\[
\mathcal{G}_\tens{S}=\tens{G}_{\tens{S}}^{-1}\mathcal{G}_{0}\tens{G}_{\tens{S}},
\]
where the tensor $\tens{G}_{\tens{S}}$ is defined in Theorem~\ref{thm:existence}. Hence, the group $\mathcal{G}_\tens{S}$ is conjugated to the group $\mathcal{G}_0$ through the tensor $\tens{G}_{\tens{S}}$, exactly as in the theory of elastic distortions \eqref{eq:sym_group_dist}.

Conversely, we now consider a strain energy of the form
\begin{equation}
\label{eq:ener_no_ISRI}
\psi(\tens{F};\,\tens{S})=f(I_1(\tens{C})-3)+g(J_1-\tr \tens{S})
\end{equation}
where $J_1 = \tr(\tens{S}\tens{C})$, and the function $f$ and $g$ must be such that the energy density \eqref{eq:ener_no_ISRI} satisfies the ISCC \eqref{eqn:ISCC} and $g$ is non-constant.  If the material is initially unstressed (i.e. $\tens{S}=\tens{0}$), the strain energy density \eqref{eq:ener_no_ISRI} defines  a general isotropic nonlinear elastic response and the material symmetry group is given by
\[
\mathcal{G}_{0}=\mathcal{O}^+(\R^3).
\]

However, if we consider an initial stress $\tens{S}=\alpha\vect{M}\otimes\vect{M}$, where $\vect{M}$ is a unit vector, we observe a change in the nature of the material symmetry group. In fact, considering that
\[
\tr(\tens{S}\tens{Q}^T\tens{C}\tens{Q})=\tr(\tens{S}\tens{C})\quad\forall\tens{F}\in\mathcal{D}\qquad\Longleftrightarrow\qquad\tens{Q}\vect{M}=\vect{M}.
\]
The material symmetry group $\mathcal{G}_{\tens{S}}$ is given by
\[
\mathcal{G}_{\tens{S}} = \left\{\tens{Q}\in\mathcal{O}^+(\R^3)\;|\;\tens{Q}\vect{M}=\vect{M}\right\},
\]
so that the material is not anymore isotropic but transversely isotropic.

Thus, if the material does not satisfy the ISRI, the material symmetry group $\mathcal{G}_\tens{S}$  is not conjugated with $\mathcal{G}_0$ and it is not possible to obtain the material symmetry group $\mathcal{G}_0$ by an elastic  distortion of the material. In other words, if the ISRI condition is not fulfilled, the  body may change its material symmetry group depending on the imposed initial stress field, leading to a modification of the material response.

\section{An illustrating example: the relaxed state of a soft disc with anisotropic initial stress}
\label{sec:7}
As an example, we consider a disc of radius $R_0$  composed of an incompressible nonlinear elastic material subjected to planar strains and initial stresses.
 Let $(\vect E_R,\,\vect{E}_\Theta)$ and $(\vect e_r,\,\vect{e}_\theta)$ be the cylindrical vector basis in Lagrangian and Eulerian coordinates respectively. We assume that the initial stress is axis-symmetric, having the following general form
\begin{equation}
\label{eq:tensSigmaDisco}
\tens{\Sigma}=\left (\alpha + \beta \log\left(\frac{R}{R_0}\right)\right)\vect E_R\otimes\vect{E}_R+\left(\gamma + \beta \log\left(\frac{R}{R_0}\right)\right)\vect E_\Theta\otimes\vect{E}_\Theta.
\end{equation}

The body in the reference configuration must obey the linear momentum balance, that in the absence of bulk forces reads
\begin{equation}
\label{eq:balance_half}
\Diver \tens{\Sigma}=\vect{0}.
\end{equation}
Since the residual stress tensor $\tens{\Sigma}$ depends only on the radial coordinate $R$, \eqref{eq:balance_half} reduces to the following scalar equation
\[
\frac{d \Sigma_{RR}}{d R} + \frac{\Sigma_{RR}-\Sigma_{\Theta\Theta}}{R}=0;
\]
that is fulfilled if and only if
\[
\beta = \gamma - \alpha.
\]
If the disc is not subjected to any external traction, then $\Sigma_{RR}(R_0)=0$, so that $\alpha=0$.
We now aim at calculating the elastic minimizer corresponding to this particular choice of the initial stresses.
Let $\psi(\tens{F};\,\tens{\Sigma})$ be the strain energy density of the initially stressed disc.  We assume that in the absence of residual stresses, the material behaves as a general isotropic material, such that
\[
\psi(\tens{F};\,\tens{0}) = f(I_1(\tens{C})-2)
\]
where $f:[0,+\infty[\rightarrow\R$ is a convex function of its scalar argument. In view of Theorem \ref{thm:uniqueness}, $\psi(\tens{F};\,\tens{\Sigma})=\psi(\tens{F}\widehat{\tens{G}}^{-1}[\tens{\Sigma}];\,\tens{0})$.
Using the polar decomposition  $\widehat{\tens{G}}[\tens{\Sigma}]= \tens{R} \tens{U}_{\widehat{\tens{G}}}$ where the tensor $\tens{R}$ is a proper orthogonal tensor and $\tens{U}_{\widehat{\tens{G}}}$ is the corresponding right stretch tensor, we denote the metric tensor of the initial elastic distortion by
\begin{equation}
\label{eq:Bg}
\tilde{\tens{B}}=\widehat{\tens{G}}[\tens{\Sigma}]^{-1}\widehat{\tens{G}}[\tens{\Sigma}]^{-T}= \tens{U}_{\widehat{\tens{G}}}^{-2},
\end{equation}
where $\lambda_{\widehat{\tens{G}}}$ is the principal eigenvalue of $\tens{U}_{\widehat{\tens{G}}}$.
Accordingly, the ISCC condition \eqref{eqn:ISCC} imposes:
\begin{equation}
\label{eq:SigmaF}
\tens{\Sigma}=2 f'(I_1(\tilde{\tens{B}})-2)\tilde{\tens{B}}-p_\tens{\Sigma}\tens{I}
\end{equation}
where $p_\tens{\Sigma}$ acts  as the Lagrange multiplier enforcing the incompressibility of the metric tensor. From the expression of the initial stress \eqref{eq:tensSigmaDisco} and the equations \eqref{eq:Bg}-\eqref{eq:SigmaF}, we get that $\tens{U}_{\widehat{\tens{G}}}$ is diagonal with respect to the cylindrical vector basis, thus
\[
\tens{U}_{\widehat{\tens{G}}}= \diag(\lambda_{\widehat{\tens{G}}}, \lambda_{\widehat{\tens{G}}}^{-1}).
\]
Considering that
\[
\tr(\widehat{\tens G}[\tens{\Sigma}]^{-T}\tens{F}^T \tens{F}\widehat{\tens G}[\tens{\Sigma}]^{-1})=\tr(\tilde{\tens{B}}\tens{C})
\]
by enforcing the ISRI condition we can write the strain energy density as
\[
\psi(\tens{F};\,\tens{\Sigma})=f(\tr(\tilde{\tens{B}}\tens{C})-2).
\]
By applying the trace and the determinant operator on both sides of \eqref{eq:SigmaF}, we obtain respectively
\begin{equation}
\label{eq:dettrSigmaF}
\left\{
\begin{aligned}
&p_\tens{\Sigma} = f'(I_1(\tilde{\tens{B}})-2)I_1(\tilde{\tens{B}})-\frac{I_1(\tens{\Sigma})}{2}\\
&I_3(\tens{\Sigma})+I_1(\tens{\Sigma})p_\tens{\Sigma} + p_\tens{\Sigma}^2 = 4 (f'(I_1(\tilde{\tens{B}})-2))^2
\end{aligned}
\right.
\end{equation}
where $I_1(\tens{\Sigma})= \tr \tens{\Sigma}$ and $I_3(\tens{\Sigma})= \det \tens{\Sigma}$.
After substituting in \eqref{eq:dettrSigmaF} the first equation into the second one, we get
\begin{equation}
\label{eq:Ib1}
\frac{I_1(\tens{\Sigma})^2}{4} - I_3(\tens{\Sigma})  = (f'(I_1(\tilde{\tens{B}})-2))^2(I_1(\tilde{\tens{B}})^2-4)
\end{equation}

The term $\frac{I_1(\tens{\Sigma})^2}{4} - I_3(\tens{\Sigma}) = \frac{(\Sigma_{RR}-\Sigma_{\Theta \Theta})^2}{4}$ is always positive. Since $f(x)$ is strictly convex with a minimum in $x=0$,  the rhs of \eqref{eq:Ib1} is a positive--definite, strictly monotone function of $I_1(\tilde{\tens{B}})$ Thus,  \eqref{eq:Ib1} is  invertible and the principal eigenvalue $\lambda_{\widehat{\tens{G}}} = \lambda_{\widehat{\tens{G}}}(\Sigma_{RR},\Sigma_{\Theta \Theta})$ is given by:
\begin{equation}
\label{eq:gsigma}
\left(\left(\lambda_{\widehat{\tens{G}}}^2+\lambda_{\widehat{\tens{G}}}^{-2}\right)^2-4\right) (f'(\lambda_{\widehat{\tens{G}}}^2+\lambda_{\widehat{\tens{G}}}^{-2}-2))^2 = \frac{\left(\Sigma_{RR}-\Sigma_{\Theta \Theta}\right)^2}{4}.
\end{equation}
We multiply each side of \eqref{eq:SigmaF} by $\tens{C}$ on the right, by applying the trace operator we get
\[
\tr(\tilde{\tens{B}}\tens{C})=\frac{J_1+ p_\tens{\Sigma}I_1}{2 f'(\lambda_{\widehat{\tens{G}}}^2+\lambda_{\widehat{\tens{G}}}^{-2}-2)}.
\]
Accordingly,  the strain energy function $\psi(\tens{F};\,\tens{\Sigma})$ for an initially stressed isotropic material reads
\begin{equation}
\label{eq:strain_ener_I1}
\psi(\tens{F};\,\tens{\Sigma})=f\left(\frac{J_1+ p_\tens{\Sigma}I_1}{2 f'(\lambda_{\widehat{\tens{G}}}^2+\lambda_{\widehat{\tens{G}}}^{-2}-2)}-2\right).
\end{equation}

Note that the relaxing map corresponding to \eqref{eq:strain_ener_I1} is the map $\widehat{\tens{G}}[\tens{\Sigma}]$ as defined by \eqref{eq:Bg} and \eqref{eq:SigmaF}, since we have written $\psi(\tens{F};\,\tens{\Sigma})$ as $\psi(\tens{F}\widehat{\tens{G}}^{-1}[\tens{\Sigma}];\,\tens{0})$.

A mapping whose deformation gradient corresponds to $\tens{U}_{\widehat{\tens{G}}}$ is given by
\begin{equation}
\label{eq:mapg}
r = \lambda_{\widehat{\tens{G}}} R,\quad\theta = \frac{\Theta}{\lambda_{\widehat{\tens{G}}}^2},\quad z = Z
\end{equation}

This relaxing map corresponds to a controllable deformation for isotropic materials, meaning that it can be supported by surface tractions alone at equilibrium. It describes the opening of the initial disc into a circular sector,  corresponding to non-homogeneous displacements and homogeneous strains \cite{singh1965note,silling1991creasing}.
In fact, we remark that  \eqref{eq:mapg} does not globally map a physically compatible  configuration even if the Riemann curvature of the underlying metric tensor is zero. This can be easily checked since the curl operator of the deformation tensor corresponding to  \eqref{eq:mapg} is not zero if $\lambda_{\widehat{\tens{G}}}\neq1$. From \eqref{eq:gsigma}, this condition implies $\Sigma_{RR}\neq\Sigma_{\Theta \Theta}$, or equivalently $\gamma\neq0$. Therefore,  the relaxing map given by \eqref{eq:mapg}  is a non-uniform controllable stress state with uniform deviatoric invariants. The latter is the necessary condition for stress controllability given in \cite{carroll1973controllable}. 

\section{Concluding remarks}

This work proved novel insights on the link between the existence of elastic minimizers and the constitutive assumptions for initially stres\-sed materials subjected to finite deformations.

Assuming a  strain energy density in the form  $\psi(\tens{F};\,\tens{\Sigma})$ and a non-degeneracy axiom,  we clarified the mathematical implications of assuming the ISRI condition as a constitutive restriction. Theorem \ref{thm:existence} proves the existence of a relaxed state given by the tensor function $\widehat{\tens{G}}[\tens{\Sigma}]$ as an implicit function of the initial stress distribution. The tensor $\widehat{\tens{G}}[\tens{\Sigma}]$ is generally not unique, and can be transformed accordingly to the symmetry group of $\psi$. Moreover, Theorem \ref{thm:uniqueness} proves that each strain energy density function $\psi(\tens{F};\,\tens{\Sigma})$ that satisfies the ISRI condition can be written as $\psi(\tens{F};\,\tens{\Sigma})=(\det\widehat{\tens{G}}[\tens{\Sigma}])\,\psi(\tens{F}\widehat{\tens{G}}[\tens{\Sigma}]^{-1};\,\tens{0})$. Thus, we prove that the material symmetry group of the initially stressed material satisfying the ISRI condition locally changes as we vary $\tens{\Sigma}$ according to the theory of elastic  distortions.

Furthermore, we have used the previous results of Ball to prove the existence Theorem~\ref{thm:existence_solution}  of the elastic minimizers for a strain energy density in the form  $\psi(\tens{F};\,\tens{\Sigma})$, that satisfies the ISRI condition under suitable constitutive restrictions. Such a result is based on the proof that the polyconvexity of the strain energy density of an initially stressed material is automatically inherited for all $\tens{\Sigma}$ if it holds in the case $\tens{\Sigma} = \tens{0}$, given some necessary conditions on the non-degeneracy and the regularity of $\psi(\cdot,\,\tens{\Sigma})$.

Whilst the theory of elastic distortion requires an a priori choice of the virtual incompatible state,  the constitutive restrictions on  $\psi(\tens{F};\,\tens{\Sigma})$  ensure the existence of the elastic minimizers corresponding to the physically observable distribution of the initial stresses. In an illustrative example, we have shown how to calculate the relaxed state of an incompressible isotropic disc as a function of the  axis-symmetric distribution of initial stresses.

We finally remark that the ISRI condition should be assumed for the materials that do not undergo a  change in the underlying  material structure, so that the initial stresses arise only in response to an elastic distortion. This happens, for example, for the residual stresses generated  by a differential growth. By using the classification proposed by  Epstein \cite{epstein2015mathematical}, the ISRI condition is indeed well suited for modeling the growth or the remodelling of a soft material, namely a change of shape that does not affect the material properties and the microstructure of the material. On the contrary, when there is a modification of the microstructure that involves a change in the material properties, the ISRI condition would be physically flawed and other constitutive choices should be done.

Our results prove useful guidelines for the constitutive restrictions on the strain energy densities of initially stresses materials, having  important applications for the study of the morphological stability and wave propagation analysis in soft tissues \cite{ciarletta2016residual}, and the non-destructive evaluation of residual stresses generated by a differential growth in biological materials \cite{li2017guided, du2018modified}.

\subsection*{Funding}
This work has been partially supported by Progetto Giovani GNFM 2017 funded by the National Group of Mathematical Physics (GNFM -- INdAM), and by the AIRC MFAG grant 17412.

\subsection*{Acknowledgements}
We are thankful to Davide Ambrosi, Giulia Bevilacqua and Alessandro Musesti for useful discussions about the contents of this paper. The authors are members of the National Group of Mathematical Physics of the Istituto Nazionale di Alta Matematics, GNFM--INdAM.


\bibliographystyle{abbrv}
\bibliography{refs}
\end{document}

%% file: defs.tex
\DeclareMathOperator{\rot}{rot}
\DeclareMathOperator{\Diver}{Div}
\DeclareMathOperator{\Grad}{Grad}
\DeclareMathOperator{\cof}{Cof}

\DeclareMathOperator{\tr}{tr}

\DeclareMathOperator{\diag}{diag}
\DeclareMathOperator{\Imag}{Im}
\DeclareMathOperator*{\argmin}{arg\,min}

\usepackage{mathtools}
\usepackage{pgfplots}
\pgfplotsset{/pgf/number format/use comma,compat=newest}

\renewcommand\L{\mathcal{L}}
\renewcommand\O{\mathcal{O}}
\renewcommand\S{\mathcal{S}}
\newcommand{\R}{\mathbb{R}}

\newcommand{\vect}[1]{\boldsymbol{#1}}
\newcommand{\tens}[1]{\mathsf{#1}}

%% file: elastic_minimizer.bbl
\begin{thebibliography}{10}

\bibitem{amar2010swelling}
M.~B. Amar and P.~Ciarletta.
\newblock Swelling instability of surface-attached gels as a model of soft
  tissue growth under geometric constraints.
\newblock {\em Journal of the Mechanics and Physics of Solids}, 58(7):935--954,
  2010.

\bibitem{ambrosi2012active}
D.~Ambrosi and S.~Pezzuto.
\newblock Active stress vs. active strain in mechanobiology: constitutive
  issues.
\newblock {\em Journal of Elasticity}, 107(2):199--212, 2012.

\bibitem{ambrosi2017solid}
D.~Ambrosi, S.~Pezzuto, D.~Riccobelli, T.~Stylianopoulos, and P.~Ciarletta.
\newblock Solid tumors are poroelastic solids with a chemo-mechanical feedback
  on growth.
\newblock {\em Journal of Elasticity}, 129(1-2):107--124, 2017.

\bibitem{ball1976convexity}
J.~M. Ball.
\newblock Convexity conditions and existence theorems in nonlinear elasticity.
\newblock {\em Archive for rational mechanics and Analysis}, 63(4):337--403,
  1976.

\bibitem{ball1984quasiconvexity}
J.~M. Ball and J.~E. Marsden.
\newblock Quasiconvexity at the boundary, positivity of the second variation
  and elastic stability.
\newblock {\em Archive for rational mechanics and analysis}, 86(3):251--277,
  1984.

\bibitem{bayly2014mechanical}
P.~Bayly, L.~Taber, and C.~Kroenke.
\newblock Mechanical forces in cerebral cortical folding: a review of
  measurements and models.
\newblock {\em Journal of the mechanical behavior of biomedical materials},
  29:568--581, 2014.

\bibitem{bilby1955continuous}
B.~Bilby, R.~Bullough, and E.~Smith.
\newblock Continuous distributions of dislocations: a new application of the
  methods of non-riemannian geometry.
\newblock {\em Proc. R. Soc. Lond. A}, 231(1185):263--273, 1955.

\bibitem{brochu2010advances}
P.~Brochu and Q.~Pei.
\newblock Advances in dielectric elastomers for actuators and artificial
  muscles.
\newblock {\em Macromolecular rapid communications}, 31(1):10--36, 2010.

\bibitem{carroll1973controllable}
M.~Carroll.
\newblock Controllable states of stress for incompressible elastic solids.
\newblock {\em Journal of Elasticity}, 3(2):147--153, 1973.

\bibitem{chuong1986residual}
C.-J. Chuong and Y.-C. Fung.
\newblock Residual stress in arteries.
\newblock In {\em Frontiers in biomechanics}, pages 117--129. Springer, 1986.

\bibitem{ciarlet2005another}
P.~G. Ciarlet and P.~Ciarlet~Jr.
\newblock Another approach to linearized elasticity and a new proof of korn's
  inequality.
\newblock {\em Mathematical Models and Methods in Applied Sciences},
  15(02):259--271, 2005.

\bibitem{ciarletta2014pattern}
P.~Ciarletta, V.~Balbi, and E.~Kuhl.
\newblock Pattern selection in growing tubular tissues.
\newblock {\em Physical review letters}, 113(24):248101, 2014.

\bibitem{ciarletta2016morphology}
P.~Ciarletta, M.~Destrade, A.~Gower, and M.~Taffetani.
\newblock Morphology of residually stressed tubular tissues: beyond the elastic
  multiplicative decomposition.
\newblock {\em Journal of the Mechanics and Physics of Solids}, 90:242--253,
  2016.

\bibitem{ciarletta2016residual}
P.~Ciarletta, M.~Destrade, and A.~L. Gower.
\newblock On residual stresses and homeostasis: an elastic theory of functional
  adaptation in living matter.
\newblock {\em Scientific reports}, 6, 2016.

\bibitem{coleman1959thermostatics}
B.~D. Coleman and W.~Noll.
\newblock On the thermostatics of continuous media.
\newblock {\em Archive for rational mechanics and analysis}, 4(1):97--128,
  1959.

\bibitem{dicarlo2002growth}
A.~DiCarlo and S.~Quiligotti.
\newblock Growth and balance.
\newblock {\em Mechanics Research Communications}, 29(6):449--456, 2002.

\bibitem{du2018modified}
Y.~Du, C.~L{\"u}, W.~Chen, and M.~Destrade.
\newblock Modified multiplicative decomposition model for tissue growth: Beyond
  the initial stress-free state.
\newblock {\em Journal of the Mechanics and Physics of Solids}, 2018.

\bibitem{ekeland1999convex}
I.~Ekeland and R.~Temam.
\newblock {\em Convex analysis and variational problems}, volume~28.
\newblock Siam, 1999.

\bibitem{epstein2015mathematical}
M.~Epstein.
\newblock Mathematical characterization and identification of remodeling,
  growth, aging and morphogenesis.
\newblock {\em Journal of the Mechanics and Physics of Solids}, 84:72--84,
  2015.

\bibitem{epstein2000thermomechanics}
M.~Epstein and G.~A. Maugin.
\newblock Thermomechanics of volumetric growth in uniform bodies.
\newblock {\em International Journal of Plasticity}, 16(7-8):951--978, 2000.

\bibitem{ericksen1954large}
J.~Ericksen and R.~Rivlin.
\newblock Large elastic deformations of homogeneous anisotropic materials.
\newblock {\em Journal of rational mechanics and analysis}, 3:281--301, 1954.

\bibitem{fichera1973existence}
G.~Fichera.
\newblock Existence theorems in elasticity.
\newblock In {\em Linear Theories of Elasticity and Thermoelasticity}, pages
  347--389. Springer, 1973.

\bibitem{gower2015initial}
A.~L. Gower, P.~Ciarletta, and M.~Destrade.
\newblock Initial stress symmetry and its applications in elasticity.
\newblock {\em Proc. R. Soc. A}, 471(2183):20150448, 2015.

\bibitem{gower2016new}
A.~L. Gower, T.~Shearer, and P.~Ciarletta.
\newblock A new restriction for initially stressed elastic solids.
\newblock {\em The Quarterly Journal of Mechanics and Applied Mathematics},
  70(4):455--478, 2017.

\bibitem{gurtin1973linear}
M.~E. Gurtin.
\newblock The linear theory of elasticity.
\newblock In {\em Linear theories of elasticity and thermoelasticity}, pages
  1--295. Springer, 1973.

\bibitem{Hoger1985}
A.~Hoger.
\newblock On the residual stress possible in an elastic body with material
  symmetry.
\newblock {\em Archive for Rational Mechanics and Analysis}, 88(3):271--289,
  1985.

\bibitem{hoger1986determination}
A.~Hoger.
\newblock On the determination of residual stress in an elastic body.
\newblock {\em Journal of Elasticity}, 16(3):303--324, 1986.

\bibitem{horgan1995korn}
C.~O. Horgan.
\newblock Korn's inequalities and their applications in continuum mechanics.
\newblock {\em SIAM review}, 37(4):491--511, 1995.

\bibitem{johnson1993dependence}
B.~E. Johnson and A.~Hoger.
\newblock The dependence of the elasticity tensor on residual stress.
\newblock {\em Journal of Elasticity}, 33(2):145--165, 1993.

\bibitem{johnson1995use}
B.~E. Johnson and A.~Hoger.
\newblock The use of a virtual configuration in formulating constitutive
  equations for residually stressed elastic materials.
\newblock {\em Journal of Elasticity}, 41(3):177--215, 1995.

\bibitem{johnson1998use}
B.~E. Johnson and A.~Hoger.
\newblock The use of strain energy to quantify the effect of residual stress on
  mechanical behavior.
\newblock {\em Mathematics and Mechanics of Solids}, 3(4):447--470, 1998.

\bibitem{kondaurov1987finite}
V.~Kondaurov and L.~Nikitin.
\newblock Finite strains of viscoelastic muscle tissue.
\newblock {\em Journal of Applied Mathematics and Mechanics}, 51(3):346--353,
  1987.

\bibitem{kroner1959allgemeine}
E.~Kr{\"o}ner.
\newblock Allgemeine kontinuumstheorie der versetzungen und eigenspannungen.
\newblock {\em Archive for Rational Mechanics and Analysis}, 4(1):273--334,
  1959.

\bibitem{Lee1969}
E.~H. Lee.
\newblock Elastic-plastic deformation at finite strains.
\newblock {\em Journal of Applied Mechanics}, 36(1):1--6, Mar 1969.

\bibitem{li2017guided}
G.-Y. Li, Q.~He, R.~Mangan, G.~Xu, C.~Mo, J.~Luo, M.~Destrade, and Y.~Cao.
\newblock Guided waves in pre-stressed hyperelastic plates and tubes:
  Application to the ultrasound elastography of thin-walled soft materials.
\newblock {\em Journal of the Mechanics and Physics of Solids}, 102:67--79,
  2017.

\bibitem{merodio2016extension}
J.~Merodio and R.~W. Ogden.
\newblock Extension, inflation and torsion of a residually stressed circular
  cylindrical tube.
\newblock {\em Continuum Mechanics and Thermodynamics}, 28(1-2):157--174, 2016.

\bibitem{morrey1952quasi}
C.~B. Morrey et~al.
\newblock Quasi-convexity and the lower semicontinuity of multiple integrals.
\newblock {\em Pacific journal of mathematics}, 2(1):25--53, 1952.

\bibitem{moulton2011circumferential}
D.~Moulton and A.~Goriely.
\newblock Circumferential buckling instability of a growing cylindrical tube.
\newblock {\em Journal of the Mechanics and Physics of Solids}, 59(3):525--537,
  2011.

\bibitem{nardinocchi2017swelling}
P.~Nardinocchi and E.~Puntel.
\newblock Swelling-induced wrinkling in layered gel beams.
\newblock {\em Proc. R. Soc. A}, 473(2207):20170454, 2017.

\bibitem{neff}
P.~Neff.
\newblock {\em Some Results Concerning the Mathematical Treatment of Finite
  Plasticity}, volume~10.
\newblock Springer Lecture Notes in Applied and Computational Mechanics, Eds.
  K. Hutter and H. Baaser, Deformation and Failure in Metallic Materials, 2013.

\bibitem{ogden2011propagation}
R.~Ogden and B.~Singh.
\newblock Propagation of waves in an incompressible transversely isotropic
  elastic solid with initial stress: Biot revisited.
\newblock {\em Journal of Mechanics of Materials and Structures},
  6(1):453--477, 2011.

\bibitem{pipkin1959formulation}
A.~Pipkin and R.~Rivlin.
\newblock The formulation of constitutive equations in continuum physics. i.
\newblock {\em Archive for Rational Mechanics and Analysis}, 4(1):129--144,
  1959.

\bibitem{rajagopal2007response}
K.~Rajagopal and A.~Srinivasa.
\newblock On the response of non-dissipative solids.
\newblock {\em Proceedings of the Royal Society of London A: Mathematical,
  Physical and Engineering Sciences}, 463(2078):357--367, 2007.

\bibitem{rajagopal2003implicit}
K.~R. Rajagopal.
\newblock On implicit constitutive theories.
\newblock {\em Applications of Mathematics}, 48(4):279--319, 2003.

\bibitem{reina2014kinematic}
C.~Reina and S.~Conti.
\newblock Kinematic description of crystal plasticity in the finite kinematic
  framework: a micromechanical understanding of f= fefp.
\newblock {\em Journal of the Mechanics and Physics of Solids}, 67:40--61,
  2014.

\bibitem{riccobelli2017shape}
D.~Riccobelli and P.~Ciarletta.
\newblock Shape transitions in a soft incompressible sphere with residual
  stresses.
\newblock {\em Mathematics and Mechanics of Solids}, page 1081286517747669,
  2017.

\bibitem{rivlin1948large}
R.~Rivlin.
\newblock Large elastic deformations of isotropic materials. ii. some
  uniqueness theorems for pure, homogeneous deformation.
\newblock {\em Phil. Trans. R. Soc. Lond. A}, 240(822):491--508, 1948.

\bibitem{rivlin1997some}
R.~Rivlin.
\newblock Some remarks concerning material symmetry.
\newblock In {\em Collected Papers of RS Rivlin}, pages 1675--1682. Springer,
  1997.

\bibitem{rivlin1948uniqueness}
R.~Rivlin and G.~Gee.
\newblock A uniqueness theorem in the theory of highly-elastic materials.
\newblock In {\em Mathematical Proceedings of the Cambridge Philosophical
  Society}, volume~44, pages 595--597. Cambridge University Press, 1948.

\bibitem{rivlin1997large}
R.~Rivlin and D.~Saunders.
\newblock Large elastic deformations of isotropic materials.
\newblock In {\em Collected papers of RS Rivlin}, pages 157--194. Springer,
  1997.

\bibitem{rivlin1987note}
R.~Rivlin and G.~Smith.
\newblock A note on material frame indifference.
\newblock {\em International journal of solids and structures},
  23(12):1639--1643, 1987.

\bibitem{rivlin1955further}
R.~S. Rivlin.
\newblock Further remarks on the stress-deformation relations for isotropic
  materials.
\newblock {\em Journal of Rational Mechanics and Analysis}, 4:681--702, 1955.

\bibitem{rivlin1997collected}
R.~S. Rivlin, G.~I. Barenblatt, and D.~D. Joseph.
\newblock {\em Collected papers of RS Rivlin}, volume~1.
\newblock Springer Science \& Business Media, 1997.

\bibitem{rivlin1955stress}
R.~S. Rivlin and J.~L. Ericksen.
\newblock Stress-deformation relations for isotropic materials.
\newblock {\em Journal of Rational Mechanics and Analysis}, 4:323--425, 1955.

\bibitem{rodriguez1994stress}
E.~K. Rodriguez, A.~Hoger, and A.~D. McCulloch.
\newblock Stress-dependent finite growth in soft elastic tissues.
\newblock {\em Journal of biomechanics}, 27(4):455--467, 1994.

\bibitem{rodriguez2016helical}
J.~Rodr{\'\i}guez and J.~Merodio.
\newblock Helical buckling and postbuckling of pre-stressed cylindrical tubes
  under finite torsion.
\newblock {\em Finite Elements in Analysis and Design}, 112:1--10, 2016.

\bibitem{shams2011initial}
M.~Shams, M.~Destrade, and R.~W. Ogden.
\newblock Initial stresses in elastic solids: constitutive laws and
  acoustoelasticity.
\newblock {\em Wave Motion}, 48(7):552--567, 2011.

\bibitem{shams2012rayleigh}
M.~Shams and R.~W. Ogden.
\newblock On rayleigh-type surface waves in an initially stressed
  incompressible elastic solid.
\newblock {\em The IMA Journal of Applied Mathematics}, 79(2):360--376, 2012.

\bibitem{shariff2017spectral}
M.~Shariff, R.~Bustamante, and J.~Merodio.
\newblock On the spectral analysis of residual stress in finite elasticity.
\newblock {\em IMA Journal of Applied Mathematics}, 82(3):656--680, 2017.

\bibitem{silling1991creasing}
S.~Silling.
\newblock Creasing singularities in compressible elastic materials.
\newblock {\em Journal of Applied Mechanics}, 58(1):70--74, 1991.

\bibitem{singh1965note}
M.~Singh and A.~C. Pipkin.
\newblock Note on ericksen's problem.
\newblock {\em Zeitschrift f{\"u}r angewandte Mathematik und Physik ZAMP},
  16(5):706--709, 1965.

\bibitem{smith1957stress}
G.~Smith and R.~S. Rivlin.
\newblock Stress-deformation relations for anisotropic solids.
\newblock {\em Archive for Rational Mechanics and Analysis}, 1(1):107--112,
  1957.

\bibitem{stylianopoulos2012causes}
T.~Stylianopoulos, J.~D. Martin, V.~P. Chauhan, S.~R. Jain, B.~Diop-Frimpong,
  N.~Bardeesy, B.~L. Smith, C.~R. Ferrone, F.~J. Hornicek, Y.~Boucher, et~al.
\newblock Causes, consequences, and remedies for growth-induced solid stress in
  murine and human tumors.
\newblock {\em Proceedings of the National Academy of Sciences},
  109(38):15101--15108, 2012.

\bibitem{taber2000modeling}
L.~A. Taber and R.~Perucchio.
\newblock Modeling heart development.
\newblock {\em Journal of Elasticity}, 61(1):165--198, 2000.

\bibitem{truesdell1965non}
C.~Truesdell and W.~Noll.
\newblock The non-linear field theories of mechanics.
\newblock {\em Handbuch der Physik}, 2:1--541, 1965.

\bibitem{wang2014modified}
H.~Wang, X.~Luo, H.~Gao, R.~Ogden, B.~Griffith, C.~Berry, and T.~Wang.
\newblock A modified holzapfel-ogden law for a residually stressed finite
  strain model of the human left ventricle in diastole.
\newblock {\em Biomechanics and modeling in mechanobiology}, 13(1):99--113,
  2014.

\end{thebibliography}
